\documentclass{article}

\usepackage{fullpage}
\usepackage{amsthm,amssymb,amsmath}  
\usepackage{xspace,enumerate}
\usepackage[utf8]{inputenc}
\usepackage{thmtools}
\usepackage{thm-restate}
\usepackage{authblk}

  \theoremstyle{plain}
  \newtheorem{theorem}{Theorem}
  \newtheorem{lemma}{Lemma}  
  \newtheorem{corollary}[theorem]{Corollary}  
  
  \newtheorem{observation}{Observation}
  \theoremstyle{definition}
  \newtheorem{definition}{Definition}
  \newtheorem{problem}{Problem}

\title{Indexing Weighted Sequences: Neat and Efficient}
\author[1]{Carl Barton}
\author[2]{Tomasz Kociumaka}
\author[3]{Chang Liu}
\author[3]{Solon P. Pissis}
\author[2,3]{Jakub Radoszewski}

\affil[1]{European Bioinformatics Institute, Wellcome Genome Campus, Hinxton, UK\\
    \texttt{carl@ebi.ac.uk}}
\affil[2]{Institute of Informatics, University of Warsaw, Warsaw, Poland\\
    \texttt{[kociumaka,jrad]@mimuw.edu.pl}}
\affil[3]{Department of Informatics, King's College London, London, UK\\
    \texttt{[chang.2.liu,solon.pissis]@kcl.ac.uk}
}

\date{\vspace{-5ex}}

  \usepackage{amsmath,amsfonts}
  \usepackage{tikz}
  \usetikzlibrary{trees}
  \usepackage{forest}  
  \usepackage[T1]{fontenc}
  \usepackage[utf8]{inputenc}
  \usepackage{comment}
  \usepackage{xspace}
  \usepackage[noend, boxed, noline]{algorithm2e}
  \usepackage{hyperref}
  \usepackage[capitalise]{cleveref}
  \usepackage{subcaption}
  \usepackage{enumerate}
 
  \tikzset{
    mynode/.style = {circle, draw, inner sep=1pt}
  }
 
    \def\dd{\mathinner{.\,.}}
    \renewcommand{\Pr}{\mathbb{P}}

  \newcommand{\floor}[1]{\left\lfloor #1 \right\rfloor}
  \newcommand{\ceil}[1]{\left\lceil #1 \right\rceil}
  \newcommand{\Oh}{\mathcal{O}}

  \def\dd{\mathinner{.\,.}}
  \newcommand{\Occpar}[1]{\mathit{Occ}_{#1}}
  \newcommand{\Occ}{\Occpar{\frac1z}}
  \newcommand{\OOcc}{\mathit{Occ}}
  \newcommand{\Count}{\mathit{Count}}
  \renewcommand{\P}{\mathit{Prob}}
  \newcommand{\MS}{\mathcal{M}}

\newcommand\twocol[2]{%
\begin{center}%
\begin{minipage}[t]{0.4\textwidth}%
\vspace{0pt}%
{#1}%
\end{minipage}\hfill%
\begin{minipage}[t]{0.6\textwidth}%
\vspace{0pt}%
{#2}%
\end{minipage}%
\end{center}}

  \newcommand{\fr}{$\frac1z$}

  \renewcommand{\SS}{\mathcal{S}}

    \newcommand{\lab}{\mathsf{L}}

\begin{document}
  \maketitle
  
\begin{abstract}
  In a {\em weighted sequence}, for every position of the sequence and every letter of the alphabet a probability 
  of occurrence of this letter at this position is specified.
  Weighted sequences are commonly used to represent imprecise or uncertain data,
  for example, in molecular biology where they are known under the name of Position-Weight Matrices.
  Given a probability threshold $\frac1z$, we say that a string $P$ of length $m$ occurs in a weighted sequence $X$ at position $i$
  if the product of probabilities of the letters of $P$ at positions $i,\ldots,i+m-1$ in $X$ is at least $\frac1z$. 
  In this article, we consider an \emph{indexing} variant of the problem, in which we are to preprocess a weighted sequence to answer
  multiple pattern matching queries.
  We present an $\Oh(nz)$-time construction of an $\Oh(nz)$-sized index for a weighted sequence of length $n$ over a constant-sized alphabet
  that answers pattern matching queries in optimal, $\Oh(m+\OOcc)$ time, where $\OOcc$ is the number of occurrences reported.
  The cornerstone of our data structure is a novel construction of a family of $\lfloor z \rfloor$ special strings
  that carries the information about all the strings that occur in the weighted sequence with a sufficient probability.
  We obtain a weighted index with the same complexities as in the most efficient previously known index by Barton et al.\ \cite{DBLP:conf/cpm/BartonKPR16},
  but our construction is significantly simpler.
  The most complex algorithmic tool required in the basic form of our index is the suffix tree
  which we use to develop a new, more straightforward index for the so-called property matching problem.
  We provide an implementation of our data structure.
  Our construction allows us also to obtain a significant improvement over the complexities of the approximate variant of the weighted index presented
  by Biswas et al.\ \cite{DBLP:conf/edbt/BiswasPTS16} and an improvement of the space complexity of their general index.
\end{abstract}

\section{Introduction} \label{sec:intro}
We consider a type of uncertain sequence called a {\em weighted sequence}.
In a weighted sequence every position contains a subset of the alphabet and every letter of the alphabet is associated with a probability of occurrence such that the sum of probabilities at each position equals 1.

Weighted sequences are common in a wide range of applications: (i) data measurements with imprecise sensor measurements; (ii) flexible sequence modelling, such as binding profiles of DNA sequences; 
(iii) observations that are private and thus sequences of observations may have artificial uncertainty introduced deliberately (see~\cite{DBLP:journals/tkde/AggarwalY09} for a survey).
Pattern matching (or substring matching) is a core operation in a wide variety of applications including genome assembly, computer virus detection, database search and short read alignment. Many of the applications of pattern matching generalise immediately to the weighted case as much of this data is more commonly uncertain (e.g.\ reads with quality scores) than certain. In particular probabilistic databases have been a very active area of research in recent years; see e.g.~\cite{DBLP:journals/cacm/DalviRS09}. 
A common assumption in practice is that the alphabet of weighted sequences is constant since the most commonly studied alphabet is $\Sigma=\{\mathtt{A},\mathtt{C},\mathtt{G},\mathtt{T}\}$.

In the \emph{Weighted Pattern Matching} (WPM) problem we are given a string $P$ called a pattern, a weighted sequence $X$ called a text, both over an alphabet $\Sigma$, and a \emph{threshold probability} \fr. The task is to find all positions $i$ in $X$ where the product of probabilities of the letters of $P$ at positions $i,\ldots,i+|P|-1$ in $X$ is at least $\frac1z$. Each such position is called an \emph{occurrence} of the pattern; we also say that the fragment and the pattern \emph{match}. 

In this article, we consider the indexing (or off-line) version of the WPM problem, called \emph{Weighted Indexing}. Here we are given a text being a weighted sequence and we are asked to construct a data structure (called an \emph{index}) to provide efficient operations for answering WPM queries related to the text.
We also consider other variants of the indexing problem.
In the {\em Approximate Weighted Indexing} problem, given a pattern and a threshold $z^{\prime}$, we are to report all occurrences of the pattern
with probability at least $\frac{1}{z^{\prime}}$ but we may also report additional occurrences with probability $\frac{1}{z^{\prime}} - \epsilon$,
for a pre-selected value of $\epsilon>0$.
In the {\em Generalised Weighted Indexing} problem we are to construct a data structure that allows for WPM queries to be answered for any threshold $z^{\prime}$ with $z^{\prime} \le z$.

A problem that is known to be closely related to the Weighted Indexing problem is \emph{Property Indexing}. 
In this problem, we are given a string $S$ called the text and a \emph{hereditary property} $\Pi$, which is a family of integer intervals contained in $\{1,\ldots,|S|\}$
(hereditary means that it is closed under subintervals). Our goal is to preprocess the text so that, for a query string $P$, we can report all occurrences of $P$ in $S$ which, interpreted as intervals, belong to $\Pi$.
The property $\Pi$ can be represented in $\Oh(|S|)$ space using an array $\pi[1\dd |S|]$ such the longest interval starting at position $i$ is $\{i,\ldots,\pi[i]\}$.

In each of the indexing problems, we denote the length of the text by $n$, the length of a query pattern by $m$, and the number of occurrences of the pattern in the text by $\OOcc$.

\subsection{Previous Results}
An $\Oh(n \log m)$-time solution for the Weighted Pattern Matching problem based on the Fast Fourier Transform was proposed in~\cite{KCL_publication,DBLP:journals/jcb/RajasekaranJS02}. Recently, an $\Oh(n \log z)$-time solution using the suffix array and lookahead scoring was presented in~\cite{DBLP:conf/isaac/KociumakaPR16}. The average case complexity of the WPM problem has also been studied and a number of fast algorithms have been shown with both linear~\cite{DBLP:journals/corr/BartonLP15} and sub-linear on average algorithms being presented~\cite{DBLP:conf/cocoa/BartonLP16}.

The Weighted Indexing problem was first considered by Iliopoulos et al.~\cite{costas_weighted_suffix_tree_j}, who introduced a data structure called {\em weighted suffix tree} allowing optimal $\Oh(m + \OOcc)$-time queries. The construction time and size of that data structure was, however, $\Oh(n |\Sigma|^{z \log z})$.

Amir et al.~\cite{amir_weighted_property_matching_j} reduced the Weighted Indexing problem to the Property Indexing problem in a text of length $\Oh(nz^2 \log z)$. For the latter, they proposed a solution with $O(n \log \log n)$ preprocessing time and optimal $O(m+\OOcc)$ query time. Later it was shown that the Property Indexing problem can be solved in linear time; see \cite{DBLP:journals/ipl/IliopoulosR08,DBLP:journals/ipl/JuanLW09} (see also \cite{DBLP:journals/tcs/Kopelowitz16}). This lead to a solution to the Weighted Indexing problem with index size and construction time $\Oh(n z^2 \log z)$, preserving optimal query time.

These results were recently improved by some of the authors in~\cite{DBLP:conf/cpm/BartonKPR16}, where they proposed an $\Oh(nz)$-sized data structure for the Weighted Indexing problem that can be constructed also in $\Oh(nz)$ time. The query time is still $O(m+\OOcc)$. The authors proposed several applications of their index.

Biswas et al.~\cite{DBLP:conf/edbt/BiswasPTS16} presented a data structure that solves the Approximate Weighted Indexing problem in $\Oh(\frac{1}{\epsilon}nz^2)$ space (with $\Omega(\frac{1}{\epsilon}n^2z^2)$ construction time) with $\Oh(m + \OOcc)$-time queries; here $\OOcc$ denotes the number of occurrences reported.
They also proposed a data structure for the Generalised Weighted Indexing problem
with $\Oh(nz^2 \log z)$ space and $\Oh(m + m \cdot \OOcc)$ query time.
The construction time is not mentioned, but a direct construction of their index works in $\Omega(n^2z^2)$ time.
Moreover, they also consider the problem of document listing for weighted sequences.

\subsection{Our Contribution}
We present a new $\Oh(nz)$-time construction of an $\Oh(nz)$-sized data structure for the Weighted Indexing problem that answers queries in optimal $\Oh(m+\OOcc)$ time.
Our index is based on a novel observation that one can always construct a family of $\lfloor z \rfloor$ special strings of length $n$ that carries all the information
about all the strings that occur in the weighted sequence.
This yields a significantly simpler construction than in the previous index \cite{DBLP:conf/cpm/BartonKPR16} preserving all of its applications.
As a by product, we obtain an optimal solution to the Property Indexing problem that avoids complex tools used in the previous solutions~\cite{amir_weighted_property_matching_j,DBLP:journals/ipl/IliopoulosR08,DBLP:journals/ipl/JuanLW09,DBLP:journals/tcs/Kopelowitz16}.
We provide a proof-of-concept implementation of our index that was validated for correctness and efficiency.
We also discuss an even simpler randomised construction with worse space complexity and construction time of the index.

Our approach lets us significantly improve upon the variants of the weighted index proposed in~\cite{DBLP:conf/edbt/BiswasPTS16}.
In the Approximate Weighted Indexing problem, we obtain $\Oh(\frac{n}{\epsilon})$ space and $\Oh(\frac{n}{\epsilon} \log{\frac{n}{\epsilon}})$ construction time, preserving the query time.
We also improve the space usage in the Generalised Weighted Indexing problem to $\Oh(nz)$, also in the document listing variant.

\subsection{Comparison of Our Techniques with the Previous Work}
Two main building blocks of our weighted index are a construction of a family of $\floor{z}$ special strings with properties
and a solution to the Property Indexing problem.

The family of strings that we construct has the same set of patterns occurring at each position as the weighted text $X$ and, moreover, the number
of occurrences of each pattern at each position is a good estimate of the probability of its occurrence at this position in $X$.
The former property is used in the construction of a weighted index and the latter in the construction of an approximate weighted index.
The existence of this family is not immediate.
However, its proof not involved and we design a $\Oh(nz)$-time elementary construction algorithm based on \emph{tries} (also known as radix trees).
In the end we show that a simple generation of a number of strings according to the probability distribution
implied by the weighted sequence with high probability yields a family of strings that also well describes the set of patterns in $X$.
However, the number of strings that one needs to generate is much larger.
Excluding the previous, exponential-size index of Iliopoulos et al.~\cite{costas_weighted_suffix_tree_j},
previous work includes the $\Oh(nz^2 \log z)$-space index of Amir et al.~\cite{amir_weighted_property_matching_j}
and $\Oh(nz)$-space index by Barton et al.~\cite{DBLP:conf/cpm/BartonKPR16}.
Amir et al.~\cite{amir_weighted_property_matching_j} show that, after a small modification of the weighted sequence, the set of maximal string patterns
that occur in it has a total length $\Oh(nz^2 \log z)$.
Barton et al.~\cite{DBLP:conf/cpm/BartonKPR16} show a representation of this set as a trie
and apply Shibuya's algorithm for suffix tree of a trie construction~\cite{DBLP:conf/isaac/Shibuya99}.

In our solution to the Property Indexing problem we construct a data structure called \emph{property suffix tree}
being the suffix tree in which the nodes corresponding to factors that do not belong to the property are trimmed.
The algorithm makes only several traversals of the suffix tree and uses an amortisation argument similar to the one from
Ukkonen's suffix tree construction~\cite{DBLP:journals/algorithmica/Ukkonen95}.
Very similar data structures were constructed by Amir et al.~\cite{amir_weighted_property_matching_j} and Kopelowitz~\cite{DBLP:journals/tcs/Kopelowitz16}.
Amir et al.~\cite{amir_weighted_property_matching_j} use a heavy machinery of weighted ancestor queries
and a fancy algorithm to mark the properties on edges of the suffix tree.
Kopelowitz~\cite{DBLP:journals/tcs/Kopelowitz16} designs an algorithm for a dynamic setting, but also mentions its static application.
He uses amortisation ideas similar to ours, but his construction is more involved due to its generality and also utilises less basic
longest common extension queries (i.e., range minimum queries).
The solution to the Property Indexing problem that was developed by Iliopoulos et al.~\cite{DBLP:journals/ipl/IliopoulosR08}
and clarified by Juan et al.~\cite{DBLP:journals/ipl/JuanLW09} constructs a different data structure that, in a sense,
shifts the hardness of the problem from the construction to the queries.
It also requires range minimum queries.

Our techniques enable us immediately to answer decision queries of a weighted index.
To answer counting and reporting queries in optimal time, we require
coloured range counting and reporting data structures in the property suffix tree
that were already used for this purpose by Barton et al.~\cite{DBLP:conf/cpm/BartonKPR16}.
In our solution to the Approximate Weighted Indexing, we need to augment the property suffix tree with
a data structure for top-$k$ document retrieval queries.
The same type of queries were used in the previous solution by Biswas et al.~\cite{DBLP:conf/edbt/BiswasPTS16},
however, not as a black box.
Moreover, they also use the less efficient reduction of~\cite{amir_weighted_property_matching_j}
which caused their data structure to use $\Oh(\frac{1}{\epsilon}nz^2)$ space, assuming that $z^{\prime} \le z$
in each query.
Finally, we improve the space complexity of the generalised weighted index of Biswas et al.~\cite{DBLP:conf/edbt/BiswasPTS16}
by plugging in our construction of $\floor{z}$ special strings.

\subsection{Structure of the Paper}
In Section~\ref{sec:z} we present a combinatorial construction of the special family of $\floor{z}$ strings.
An efficient implementation of the construction of this family based on tries is proposed in Section~\ref{sec:constant}.
In Section~\ref{sec:index} the new optimal solution for the Property Indexing problem is presented.
Using the construction and the property index, we obtain our weighted index in Section~\ref{sec:appl}
and, with the aid of an auxiliary tool, an approximate weighted index in Section~\ref{sec:appl2}.
Alternative randomised constructions of the two indexes with worse parameters are discussed in Section~\ref{sec:prob}.
Our improvement to the Generalised Weighted Index and our C++ implementation are briefly discussed Section~\ref{sec:concl}.

  \section{Preliminaries}\label{sec:prelim}
  \subsection{Strings and Property Indexing}
  A \emph{string} $S$ over an alphabet $\Sigma$ is a finite sequence of letters from $\Sigma$.
  By $n=|S|$ we denote the length of $S$ and by $S[i]$, for $1 \le i \le n$, we denote the $i$-th letter of $S$.
  By $S[i{\dd}j]$ we denote the string $S[i] \ldots S[j]$ called a \emph{factor} of $S$ (if $i>j$, then the factor is an empty string).
  A factor is called a \emph{prefix} if $i=1$ and a \emph{suffix} if $j=n$.
  We say that a string $P$ \emph{occurs} at position $i$ in $S$ if $P=S[i \dd i+|P|-1]$.

  A \emph{property} $\Pi$ of $S$ is a hereditary collection of integer intervals contained in $\{1,\ldots,n\}$.
  For simplicity, we represent every property $\Pi$ with an array $\pi[1\dd |S|]$ such that the longest interval $I\in \Pi$
  starting at position $i$ is $\{i,\ldots,\pi[i]\}$.
  Observe that $\pi$ can be an arbitrary array satisfying $\pi[i] \in \{i-1,\ldots,n\}$  and $\pi[1] \le \pi[2] \le \cdots \le \pi[n]$.
  For a string $P$, by $\OOcc_\pi(P,S)$ we denote the set of occurrences $i$ of $P$ in $S$ such that $i+|P|-1 \le \pi[i]$.
  These notions lead us to the statement of the following problem.

  \begin{problem}[Property Indexing]~\\
    \textbf{Input:} A string $S$ of length $n$ over an alphabet $\Sigma$ and an array $\pi$ representing a property $\Pi$.\\
    \textbf{Queries:} For a given pattern string $P$ of length $m$, compute $|\OOcc_\pi(P,S)|$ or report all elements of $\OOcc_\pi(P,S)$.
  \end{problem}
  Let us consider an indexed family $\SS=(S_j,\pi_j)_{j=1}^k$ of strings $S_j$ with properties $\pi_j$.
  For a string $P$ and an index $i$, by
  $$\Count_\SS(P,i)=|\{j:i \in \OOcc_{\pi_j}(P,S_j)\}|$$
  we denote the total number of occurrences of $P$ at the position $i$
  in the strings $S_1,\ldots,S_k$ that respect the properties.

  \subsection{Weighted Sequences and Weighted Indexing}
  A weighted sequence $X=x_1x_2 \ldots x_n$ of length $|X|=n$ over an alphabet $\Sigma$
  is a sequence of sets of pairs of the form
  $x_i = \{(c,\ p^{(X)}_i(c))\ :\ c \in \Sigma\}$.
  Here, $p^{(X)}_i(c)$ is the occurrence probability of the letter $c$ at the position $i \in \{1,\ldots,n\}$.
  These values are non-negative and sum up to 1 for a given $i$.
  An example of a weighted sequence is shown in Table~\ref{tab:WS}.

  \begin{table}[htpb]
      \begin{center}
      \begin{tabular}{|r|p{0.4cm}|p{0.4cm}|p{0.4cm}|p{0.4cm}|p{0.4cm}|p{0.4cm}|}
        \hline
        $i$ &1&2&3&4&5&6\\\hline
        $p_i^{(X)}(\mathtt{A})$ & 1 & $\tfrac12$ & $\tfrac34$ & $\tfrac45$ & $\tfrac12$ & $\tfrac14$ \\\hline
        $p_i^{(X)}(\mathtt{B})$ & 0 & $\tfrac12$ & $\tfrac14$ & $\tfrac15$ & $\tfrac12$ & $\tfrac34$ \\\hline
      \end{tabular}
      \end{center}
    \caption{
      A weighted sequence $X$ of length 6 over $\Sigma=\{\mathtt{A},\mathtt{B}\}$.
        \label{tab:WS}
    }
  \end{table}
  
  The \emph{probability of matching} of a string $P$ at position $i$ of a weighted sequence $X$ equals
  $$\P_X(P,i) = \prod_{j=1}^{|P|} p^{(X)}_{i+j-1}(P[j]).$$
  We say that a string $P$ \emph{occurs} in $X$ at position $i$ if $\P_X(P,i) \ge \frac1z$.
  We also say that $P$ is a \emph{solid factor} of $X$ (starting, occurring) at position $i$.
  By $\Occ(P,X)$ we denote the set of all positions where $P$ occurs in $X$.
  The main problem in scope can be formulated as follows.
  \begin{problem}[Weighted Indexing]~\\
    \textbf{Input:} A weighted sequence $X$ of length $n$ over an alphabet $\Sigma$ and a threshold \fr.\\
    \textbf{Queries:} For a given pattern string $P$ of length $m$, check if $\Occ(P,X) \ne \emptyset$ (\emph{decision} query),
    compute $|\Occ(P,X)|$ (\emph{counting} query), or report all elements of $\Occ(P,X)$ (\emph{reporting} query).
  \end{problem}

  \paragraph{\bf Our model of computations.}
  We assume the word-RAM model with word size $w = \Omega(\log(nz))$.
  We consider the log-probability model of representations of weighted sequences in which probabilities can be multiplied exactly in $\Oh(1)$ time.
  We further assume that $|\Sigma| = \Oh(1)$;
  under this assumption a weighted sequence of length $n$ has a representation using $\Oh(n)$ space.

\section{Existence of an Equivalent Family of Strings}\label{sec:z}
In the definition below, we formalise the property of a string family that we aim to construct.
\begin{definition}\label{def:estim}
  We say that an indexed family $\SS=(S_j,\pi_j)_{j=1}^{\floor{z}}$ containing strings $S_j$ of length $n$
  is a \emph{$z$-estimation} of a weighted sequence $X$ of length $n$ if and only if, for every string $P$ and position $i \in \{1,\ldots,n\}$,
  $\Count_\SS(P,i) = \floor{\P_X(P,i)z}$.
\end{definition}
Note that a $z$-estimation $\SS$ of a weighted sequence $X$ carries the information about all solid factors of $X$:
a string $P$ occurs in $X$ at position $i$ if and only if it occurs at position $i$
in at least one of the strings $S_j$ respecting its property $\pi_j$.
This observation will be used in the construction of our weighted index.
Moreover, the value $\Count_\SS(P,i)$ provides a good estimation of the probability $\P_X(P,i)$:
$$\tfrac{1}{z} \Count_\SS(P,i) \le \P_X(P,i) < \tfrac{1}{z} \Count_\SS(P,i) + \tfrac{1}{z}.$$
This will let us design an approximate weighted index.
An example of a $z$-estimation is shown in Table~\ref{tab:S}.

  \begin{table}[htpb]
  \twocol{
      \begin{center}
      \begin{tabular}{|r|l|l|l|l|l|l|}
        \hline
        $i$ &1&2&3&4&5&6\\\hline\hline
        $S_1[i]$ & $\mathtt{A}$ & $\mathtt{A}$ & $\mathtt{A}$ & $\mathtt{A}$ & $\mathtt{A}$ & $\mathtt{A}$ \\\hline
        $\pi_1[i]$ & 2 & 2 & 3 & 4 & 5 & 6 \\\hline\hline
        
        $S_2[i]$ & $\mathtt{A}$ & $\mathtt{A}$ & $\mathtt{A}$ & $\mathtt{A}$ & $\mathtt{A}$ & $\mathtt{B}$ \\\hline
        $\pi_2[i]$ & 4 & 4 & 5 & 6 & 6 & 6 \\\hline\hline
                
        $S_3[i]$ & $\mathtt{A}$ & $\mathtt{B}$ & $\mathtt{A}$ & $\mathtt{A}$ & $\mathtt{B}$ & $\mathtt{B}$ \\\hline
        $\pi_3[i]$ & 4 & 4 & 5 & 6 & 6 & 6 \\\hline\hline

        $S_4[i]$ & $\mathtt{A}$ & $\mathtt{B}$ & $\mathtt{B}$ & $\mathtt{B}$ & $\mathtt{B}$ & $\mathtt{B}$ \\\hline
        $\pi_4[i]$ & 2 & 2 & 3 & 3 & 5 & 6 \\\hline
      \end{tabular}
      \end{center}
  }{
      \begin{center}
      \begin{tabular}{c|c|c}
        string $P$ & $\P_X(P,3)$ & $\{j : 3 \in \OOcc_{\pi_j}(P,S_j)\}$ \\\hline
        $\varepsilon$ & 1 & 1, 2, 3, 4 \\\hline
        $\mathtt{A}$ & $0.75$ & 1, 2, 3 \\\hline
        $\mathtt{AA}$ & $0.6$ & 2, 3 \\\hline
        $\mathtt{AAA}$ & $0.3$ & 2 \\\hline
        $\mathtt{AAB}$ & $0.3$ & 3 \\\hline
        $\mathtt{B}$ & $0.25$ & 4 \\\hline
      \end{tabular}
      \end{center}
      }
      \caption{
        To the left: a 4-estimation $\SS$ of the weighted sequence $X$ from Table~\ref{tab:WS}.
        To the right: all the strings that occur at position $i=3$ in $X$
        together with the probabilities of occurrence in $X$ and occurrences in $\SS$.
      \label{tab:S}
      }
  \end{table}

Below, we prove existence of a $z$-estimation. An efficient construction is deferred to the next section.

For a fixed weighted sequence $X$ of length $n$ and a threshold $z$, we can use compact notation:
$$t_i(P)=\floor{\P_X(P,i)z} \quad\mbox{and}\quad m_i(P)=t_i(P)-\sum_{c\in \Sigma} t_i(Pc)$$
for $i=1,\ldots,n$.
We start with an equivalent characterisation of $z$-estimations of $X$.
\begin{observation}\label{obs:ms}
A family $\SS=(S_j,\pi_j)_{j=1}^{\floor{z}}$ is a $z$-estimation of $X$
if and only if for each position $i$, every string $P$ is a prefix of exactly $t_i(P)$ strings $S_j[i\dd \pi_j[i]]$.
\end{observation}
Next, we prove that this condition uniquely defines the multiset $\{S_j[i\dd \pi_j[i]] : 1\le j \le \floor{z}\}$.
\begin{lemma}
There exists a unique multiset $\MS_i$ such that each string $P$ is a prefix of exactly $t_i(P)$ strings in $\MS_i$.
\end{lemma}
\begin{proof}
Consider a multiset $\MS_i$ satisfying the required condition and an arbitrary string $P$.
For each $c\in \Sigma$, there are $t_i(Pc)$ strings in $\MS_i$ with the prefix $P$ is followed by a character $c$.
In the remaining $t_i(P)-\sum_{c\in \Sigma} t_i(Pc)$ strings in $\MS_i$, the prefix $P$ it is not followed by any letter.
Thus, the multiplicity of $P$ in $\MS_i$ must be $m_i(P)$. This implies uniqueness of $\MS_i$.

Observe that $t_i(P)\ge \sum_{c\in \Sigma} t_i(Pc)$, because  $\P_X(P,i)\ge \sum_{c\in \Sigma} \P_X(Pc,i)$ and the function $x \mapsto \floor{xz}$ is superadditive.
Consequently, we may define a multiset $\MS_i$ using values $m_i(P)$ as multiplicities. It remains to prove that this multiset satisfies the required condition.
For this, we consider strings $P$ in the order of decreasing lengths. The base case is trivial because strings $P$ longer than $X$ 
satisfy $\P_X(P,i)=0$.
The inductive hypothesis yields that, for each $c\in \Sigma$, the string $Pc$ is a prefix of
$t_i(Pc)$ strings in $\MS_i$. Consequently, the string $P$ is a prefix of $m_i(P)+\sum_{c\in \Sigma} t_i(Pc) = t_i(P)$ strings in $\MS_i$, as claimed.
\end{proof}

Observe that in a $z$-estimation, $S_{j}[i\dd \pi_j[i]]$ can be obtained from $S_{j}[i+1\dd \pi_j[i+1]]$
by inserting a leading character and dropping some number of trailing characters.
This statement holds if only $\pi_j[i] \ge i$; otherwise $S_{j}[i\dd \pi_j[i]]=\varepsilon$.
The relation between these strings can be formalised as follows:
\begin{definition}
We say that $P\in \MS_i$ is \emph{compatible} with $Q\in \MS_{i+1}$ if $P=\varepsilon$ or $P=cQ'$ for some character $c\in \Sigma$
and a prefix $Q'$ of $Q$.
\end{definition}

Thus, if a $z$-estimation exists, it yields a perfect matching between $\MS_{i+1}$ and $\MS_{i}$
such that the matched strings are compatible.
We prove that such a matching exists unconditionally.
For an example, see Table~\ref{tab:match}.

  \begin{table}[htpb]
      \begin{center}
      \begin{tabular}{ccccccccccc}
        $\MS_1$ && $\MS_2$ && $\MS_3$ && $\MS_4$ && $\MS_5$ && $\MS_6$\\
        $\mathtt{\underline{A}A}$ & --- & $\mathtt{\underline{A}}$ & --- & $\mathtt{\underline{A}}$ & --- & $\mathtt{\underline{A}}$ & --- & $\mathtt{\underline{A}}$ & --- & $\mathtt{\underline{A}}$ \\
        $\mathtt{\underline{A}AAA}$ & --- & $\mathtt{\underline{A}AA}$ & --- & $\mathtt{\underline{A}AA}$ & --- & $\mathtt{\underline{A}AB}$ & --- & $\mathtt{\underline{A}B}$ & --- & $\mathtt{\underline{B}}$ \\
        $\mathtt{\underline{A}BAA}$ & --- & $\mathtt{\underline{B}AA}$ & --- & $\mathtt{\underline{A}AB}$ & --- & $\mathtt{\underline{A}BB}$ & --- & $\mathtt{\underline{B}B}$ & --- & $\mathtt{\underline{B}}$ \\
        $\mathtt{\underline{A}B}$ & --- & $\mathtt{\underline{B}}$ & --- & $\mathtt{\underline{B}}$ & --- & $\varepsilon$ & --- & $\mathtt{\underline{B}}$ & --- & $\mathtt{\underline{B}}$
      \end{tabular}
      \end{center}
      \caption{
        The sets $\MS_i$ for the weighted sequence $X$ from Table~\ref{tab:WS} with $z=4$.
        Perfect matchings of compatible strings between $\MS_i$ and $\MS_{i+1}$ are marked.
        The first letters of the strings form the 4-estimation from Table~\ref{tab:S}
        and the length of the $j$-th string in $\MS_i$ corresponds to $\pi_j[i]-i+1$.
      \label{tab:match}
      }
  \end{table}

\begin{lemma}\label{lem:compat}
For every $1 \le i \le n-1$, there is a one-to-one correspondence from $\MS_{i+1}$ into $\MS_i$ such that each $Q\in \MS_{i+1}$ is matched with a compatible $P\in \MS_i$.
\end{lemma}
\begin{proof}
We greedily transform each $Q\in \MS_{i+1}$ into the longest compatible $P\in \MS_{i}$ which is still unmatched.
If no compatible $P\in \MS_{i}$ is available, we leave $Q$ unmatched.
We will show that all strings $Q\in \MS_{i+1}$ are actually matched at the end of this process.
Since $|\MS_{i}|=t_i(\varepsilon)=\floor{z}=t_{i+1}(\varepsilon)=|\MS_{i+1}|$, it suffices to prove that no $P\in \MS_i$ is left unmatched.

An empty string $P\in \MS_i$ is compatible with every $Q\in \MS_{i+1}$, so it cannot be left unmatched.
Thus, suppose that $P=cQ' \in \MS_{i}$, for some $c \in \Sigma$ and string $Q'$, is left unmatched.
Let us denote by $\mathcal{R}$ the multiset containing all strings $Q\in \MS_{i+1}$ compatible with $P$, i.e., starting with $Q'$.
We further define $\mathcal{L}$ as the multiset containing all strings $P'\in \MS_{i}$ that start with $c'Q'$ for some $c' \in \Sigma$.
The construction procedure guarantees that each $Q\in \mathcal{R}$ has been matched to
a compatible $P'$ satisfying $|P'|\ge |P|$; such $P'$ must belong to the multiset $\mathcal{L}$.

Observe that $|\mathcal{L}| = \sum_{c'\in \Sigma} t_i(c'Q') \le t_{i+1}(Q')= |\mathcal{R}|$ because $\P_X(Q',i+1)\ge \sum_{c'\in \Sigma} \P_X(c'Q',i)$ and the function $x \mapsto \floor{xz}$ is superadditive. Consequently, each $P'\in \mathcal{L}$ must be matched to some $Q\in \mathcal{R}$.
Since $P\in \mathcal{L}$ is unmatched, we obtain a contradiction.
\end{proof}

Due to \cref{lem:compat}, we can index the strings $\MS_{i}=\{P_{j,i} : 1\le j \le \floor{z}\}$ so
that we have $\floor{z}$ chains $P_{j,1},\ldots,P_{j,n},P_{j,n+1}=\varepsilon$ with compatible subsequent strings.
It is easy to transform each such chain to a string $S_j$ with property $\pi_j$ so that 
$S_j[i\dd \pi_j[i]] = P_{j,i}$.
The value $S_j[i]$ is not specified if $P_{j,i}=\varepsilon$;
in this case, we may set $S_j[i]$ to an arbitrary character.
The resulting family $\SS = (S_j,\pi_j)_{j=1}^{\floor{z}}$ clearly satisfies
the characterisation of \cref{obs:ms}, which completes the proof of the following result.
\begin{theorem}\label{thm:z_cover}
  Each weighted sequence $X$ has a $z$-estimation.
\end{theorem}

\section{Efficient Implementation}\label{sec:constant}
In this section we describe an algorithm which, given a weighted sequence $X$ of length $n$ and threshold $z$,
constructs a $z$-estimation of $X$ in $\Oh(nz)$ time.

At a high level, we follow the existential construction of \cref{sec:z}. 
We start with $\MS_{n+1}$, which consists of $\floor{z}$ copies of $\varepsilon$,
and we iterate over positions $i=n,\ldots,1$ transforming $\MS_{i+1}$ to $\MS_i$ so that each $P_{j,i+1}\in \MS_{i+1}$
is replaced with a compatible string $P_{j,i}\in \MS_i$. We simultaneously build the $z$-estimation $\SS = (S_j,\pi_j)_{j=1}^{\floor{z}}$.
More precisely, we set $\pi_j[i]$ to $i+|P_{j,i}|-1$ and $S_j[i]$ to the leading character of $P_{j,i}$, or an arbitrary character if $P_{j,i}=\varepsilon$.

Each transformation simulates the procedure provided in the proof of \cref{lem:compat}.
However, our implementation uses \emph{solid factor tries} in order to achieve $\Oh(z)$ amortised running time.

\subsection{Solid Factor Tries}
Recall that a trie is a rooted tree in which each node represents a string; the string corresponding to node $u$, called the \emph{label} of $u$, is denoted $\lab(u)$.
The root has label $\varepsilon$, and the parent of a node $u$ with $\lab(u)=Pc$ for $c\in \Sigma$
is the node $v$ with $\lab(v)=P$; the edge from $P$ to $Pc$ is \emph{labelled} with $c$.
Observe that the family of solid factors occurring at position $i$ (i.e., strings $P$ such that $t_i(P)>0$)
is closed with respect to prefixes. Thus, we can define a \emph{solid factor trie} $T_i$ whose nodes represent these factors.

We store $\MS_i$ using \emph{tokens} in $T_i$: each $P_{j,i}\in \MS_{i}$ is represented by a token (with identifier~$j$)
located at the node $u\in T_i$ with $\lab(u)=P_{j,i}$. For each token $j$, we store the node $u\in T_i$ with $\lab(u)=P_{j,i}$ and the probability $\P_X(P_{j,i},i)$. 
Observe that the number of tokens at the node $u$ is $m_i(\lab(u))$ and the number of tokens in the subtree rooted at $u$ is $t_i(\lab(u))$.
To simplify notation, we denote $m_i(u)=m_i(\lab(u))$ and $t_i(u)=t_i(\lab(u))$. We have the following simple observation; see also Figure~\ref{fig:sftrees}.

\begin{observation}
  The trie $T_i$ contains $\floor{z}$ tokens in total and every leaf contains tokens.
\end{observation}

\begin{figure}[htpb]
  \begin{center}
    \begin{tikzpicture}[scale=0.4]
      \begin{scope}
        \node[mynode] (X) {\,}
        child {
          node[mynode] {\,}
          child {
            node[mynode] {\footnotesize 1}
            child {
              node[mynode] {\,}
              child {
                node[mynode] {\footnotesize 2}
                edge from parent node[left] {\footnotesize$\mathtt{A}$}
              }
              edge from parent node[left] {\footnotesize$\mathtt{A}$}
            }
            edge from parent node[left] {\footnotesize$\mathtt{A}$}
          }
          child {
            node[mynode] {\footnotesize 4}
            child {
              node[mynode] {\,}
              child {
                node[mynode] {\footnotesize 3}
                edge from parent node[right] {\footnotesize$\mathtt{A}$}
              }
              edge from parent node[right] {\footnotesize$\mathtt{A}$}
            }
            edge from parent node[right] {\footnotesize$\mathtt{B}$}
          }
          edge from parent node[left] {\footnotesize$\mathtt{A}$}
        };
        \draw (X) node[below=2.7cm] {$T_1$};
      \end{scope}

      \begin{scope}[xshift=6cm]
        \node[mynode] (Y) {\,}
        child {
          node[mynode] {\footnotesize 1}
          child {
            node[mynode] {\,}
            child {
              node[mynode] {\footnotesize 2}
              edge from parent node[left] {\footnotesize$\mathtt{A}$}
            }
            edge from parent node[left] {\footnotesize$\mathtt{A}$}
          }
          edge from parent node[left] {\footnotesize$\mathtt{A}$}
        }
        child {
          node[mynode] {\footnotesize 4}
          child {
            node[mynode] {\,}
            child {
              node[mynode] {\footnotesize 3}
              edge from parent node[right] {\footnotesize$\mathtt{A}$}
            }
            edge from parent node[right] {\footnotesize$\mathtt{A}$}
          }
          edge from parent node[right] {\footnotesize$\mathtt{B}$}
        };
        \draw (Y) node[below=2.7cm] {$T_2$};
      \end{scope}

      \begin{scope}[xshift=12cm]
        \node[mynode] (Z) {\,}
        child {
          node[mynode] {\footnotesize 1}
          child {
            node[mynode] {\,}
            child {
              node[mynode] {\footnotesize 2}
              edge from parent node[left] {\footnotesize$\mathtt{A}$}
            }
            child {
              node[mynode] {\footnotesize 3}
              edge from parent node[right] {\footnotesize$\mathtt{B}$}
            }
            edge from parent node[left] {\footnotesize$\mathtt{A}$}
          }
          edge from parent node[left] {\footnotesize$\mathtt{A}$}
        }
        child {
          node[mynode] {\footnotesize 4}
          edge from parent node[right] {\footnotesize$\mathtt{B}$}
        };
        \draw (Z) node[below=2.7cm] {$T_3$};
      \end{scope}

      \begin{scope}[xshift=18cm]
        \node[mynode] (X1) {\footnotesize 4}
        child {
          node[mynode] {\footnotesize 1}
          child {
            node[mynode] {\,}
            child {
              node[mynode] {\footnotesize 2}
              edge from parent node[left] {\footnotesize$\mathtt{B}$}
            }
            edge from parent node[left] {\footnotesize$\mathtt{A}$}
          }
          child {
            node[mynode] {\,}
            child {
              node[mynode] {\footnotesize 3}
              edge from parent node[right] {\footnotesize$\mathtt{B}$}
            }
            edge from parent node[right] {\footnotesize$\mathtt{B}$}
          }
          edge from parent node[left] {\footnotesize$\mathtt{A}$}
        };
        \draw (X1) node[below=2.7cm] {$T_4$};
      \end{scope}

      \begin{scope}[xshift=24cm]
        \node[mynode] (Y1) {\,}
        child {
          node[mynode] {\footnotesize 1}
          child {
            node[mynode] {\footnotesize 2}
            edge from parent node[left] {\footnotesize$\mathtt{B}$}
          }
          edge from parent node[left] {\footnotesize$\mathtt{A}$}
        }
        child {
          node[mynode] {\footnotesize 4}
          child {
            node[mynode] {\footnotesize 3}
            edge from parent node[right] {\footnotesize$\mathtt{B}$}
          }
          edge from parent node[right] {\footnotesize$\mathtt{B}$}
        };
        \draw (Y1) node[below=2.7cm] {$T_5$};
      \end{scope}

      \begin{scope}[xshift=30cm,level 1/.style = {sibling distance=3cm}]
        \node[mynode] (Z1) {\,}
        child {
          node[mynode] {\footnotesize 1}
          edge from parent node[left] {\footnotesize$\mathtt{A}$}
        }
        child {
          node[ellipse, draw, inner sep=1pt] {\footnotesize 2,3,4}
          edge from parent node[right] {\footnotesize$\mathtt{B}$}
        };
        \draw (Z1) node[below=2.7cm] {$T_6$};
      \end{scope}

    \end{tikzpicture}
  \end{center}
\caption{
  The solid factor tries for the weighted sequence $X$ from Table~\ref{tab:WS} with $z=4$.
  Tokens in the nodes are numbered according to the order from Table~\ref{tab:match}.
  \label{fig:sftrees}
}
\end{figure}
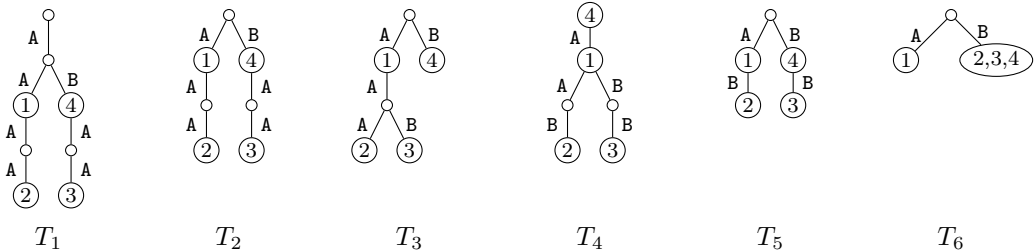

\subsection{Transformation Algorithm}
For each index $i$, we transform the solid factor trie $T_{i+1}$ to $T_{i}$
and move the tokens so that $\MS_{i+1}$ is transformed to $\MS_i$.

Before we describe the implementation, let us formulate a relation between $T_i$ and $T_{i+1}$.
\begin{observation}\label{obs:rel}
If $u\in T_{i}$ has a non-empty label, $\lab(u)=cP$, for some $c\in \Sigma$,
then $T_{i+1}$ contains a node $v$ with label $\lab(v)=P$.
\end{observation}
Consequently, each non-root node $u\in T_i$ has a corresponding node $v\in T_{i+1}$. In our construction algorithm,
we sometimes reuse $v$ as $u$; otherwise, we create $u$ as a copy of $v$.
More precisely, we distinguish a \emph{heavy letter} $h\in \Sigma$ maximising probability $p^{(X)}_i(c)$ over $c\in \Sigma$.
We reuse $v$ if $\lab(u)$ starts with $h$ and create a copy of $v$ otherwise.

This approach is implemented as follows. First, we create the root of $T_i$ and attach $T_{i+1}$ to the new root using an edge with label $h$. 
The resulting subtree, denoted $T_{i,h}$, contains all tokens present in $T_{i+1}$ and may contain nodes $v$ with $t_i(v)=0$ (we piggyback trimming them to the last phase when tokens are moved).
Next, we consider all the remaining letters $c\in \Sigma\setminus\{h\}$. For each such letter we shall build a subtree $T_{i,c}$ representing
solid factors occurring at position $i$ and starting with character $c$.
We simultaneously build and traverse $T_{i,c}$: we construct the children of a node $u$ while visiting $u$ for the first time.
While at node $u$ with $\lab(u)=cP$, we maintain the probability $\P_X(cP,i)$ and a pointer to the corresponding node $v\in T_{i,h}$
such that $\lab(v)=hP$. To construct the children of $u$, we simply compute $t_i(cPc')$ for each $c'\in \Sigma$.
Moreover, we determine $m_i(cP)$ and place $m_i(cP)$ \emph{token requests} at node $v$, announcing that $m_i(cP)$ tokens are needed at $u$.

\begin{figure}[htpb]
\include{__fig_transform}
\caption{
  Transformation between $T_4$ to $T_3$ from the example in Figure~\ref{fig:sftrees}.
  To the left: the trie $T_4$ with letter probabilities (in blue).
  In the middle: the trie $T_4$ is copied as $T_{3,\mathtt{A}}$, whereas $T_{3,\mathtt{B}}$ is created using a backtracking algorithm
  (in this case, it has only one node).
  Asterisks denote nodes that require tokens.
  The token request is shown with an arrow.
  To the right: the final $T_3$ created after the tokens are moved up and redundant nodes are removed.
  Note that the tokens number 1 and 4 could have been interchanged depending on the order of processing.
  \label{fig:transform}
}
\end{figure}

Finally, we move the tokens and trim the redundant nodes of $T_{i,h}$.
We process the tokens in an arbitrary order. Consider a token located at node $v$ of $T_{i,h}$ with $\lab(v)=hQ$ (the token used to represent $Q\in \MS_{i+1}$).
We traverse the path from $v$ towards the root of $T_i$ maintaining the probability $\P_X(\lab(v'),i)$ at the currently visited node $v'$.
First, we check if there is any token request at $v'$. If so, we comply with the request, remove it, and terminate the traversal. 
Otherwise, we compute $m_i(v')$ using the probability. If $v'$ contains less than $m_i(v')$ already processed tokens,
we place our token at $v'$ and terminate the traversal.
Otherwise, we proceed to the parent of $v'$.  If $v'$ is a leaf and does not contain any (processed or unprocessed) tokens, we remove $v'$ from $T_{i,h}$.
If the traversal reaches the root of $T_i$, we place the token unconditionally at the root.
Figure~\ref{fig:transform} illustrates this procedure on an example.

\subsubsection{Correctness}
We shall prove that the procedure described above correctly computes $T_i$ and $\MS_i$.
Due to \cref{obs:rel}, the trie $T_i$ contains all the necessary nodes. We only need to prove 
that no redundant nodes $v$ (with $t_i(v)=0$) are left in $T_{i,h}$.
Suppose that $v$ is the deepest such node; clearly, it must be a leaf of $T_{i,h}$.
We did not place the token at $v$ because $m_i(v)\le t_i(v)=0$. On the other hand, tokens were present in all leaves of $T_{i+1}$,
so the subtree of $v$ in $T_{i,h}$ initially contained a token. Let us consider the moment of moving the last token in this subtree.
If the token travelled further to the parent of $v$, we would have removed $v$.
Thus, the token must have been placed at a node $u$ complying with a token request at $u$.
However, in that case we have $t_i(v)\ge t_i(u) \ge m_i(u)>0$, because $h$ is the heavy character. 
This contradiction concludes the proof.

Hence, we proceed to proving that the final configuration of tokens represents $\MS_i$.
For this, we observe that our algorithm simulates the greedy procedure in the proof of \cref{lem:compat}.
In other words, we shall prove that we transformed $P_{j,i+1}\in \MS_{i+1}$ 
to the longest compatible element of $\MS_i$ which was still unmatched when we processed token $j$.
Suppose that there was an unmatched string $P'\in \MS_i$ longer than $P_{j,i}$.
Let $P'=cQ'$ and observe that, when processing token $j$, we visited the node $v'$ with $\lab(v')=hQ'$.
If $c= h$, then we would have less than $m_i(v')$ processed tokens at $v'$.
Otherwise, there must have been a token request at $v'$.
For either event we would not have proceeded to the parent of $v'$.
This contradiction concludes the proof.

\subsubsection{Running Time Analysis}

It remains to show that the total running time of the $n$ transformations is $\Oh(nz)$. 
In a single iteration, processing the $j$-th token, i.e., transforming $P_{j,i+1}$ to $P_{j,i}$,
we visited at most $1+|P_{j,i+1}|-|P_{j,i}|$ nodes of $T_{i,h}$ and deleted some of them.
Across all iterations this is $\Oh(n)$ per token and $\Oh(nz)$ in total.
The remaining operations (construction of subtrees $T_{i,c}$) take $\Oh(1)$ time per created node.
The final tree $T_1$ has $\Oh(nz)$ nodes and the overall number of deleted nodes is $\Oh(nz)$.
Hence, the total number of created nodes is also $\Oh(nz)$.

This concludes the proof that the running time is $\Oh(nz)$.
Hence, we achieve the main goal of this section.
\begin{theorem}\label{thm:specweight}
For a weighted sequence $X$ of length $n$ over a constant-sized alphabet, one can construct a $z$-estimation in $\Oh(nz)$ time.
\end{theorem}

\newcommand{\suflink}{\mathsf{suflink}}
\newcommand{\eps}{\varepsilon}

\section{Property Indexing Made Simple}\label{sec:index}
  Every known solution to the Property Indexing problem makes use of suffix trees; ours is no exception.
  Below we recall the basics on suffix trees.
  \subsection{Suffix Trees}
  The \textit{suffix tree} $T$ of a non-empty string $S$ of length $n$ is a compact trie representing 
  all suffixes of $S$. The nodes of the trie which become nodes of the suffix tree
  (i.e., branching nodes, terminal nodes, and the root) are called {\it explicit} nodes, 
  while the other nodes are called {\it implicit}.
  The edges out-going from a node are labelled with their first letters and can be stored, e.g., in a list.
 
  Each edge of the suffix tree can be viewed as an upward maximal 
  path of implicit nodes starting with an explicit node. Moreover, each node belongs to a unique path of that kind. 
  Then, each node of the trie can be represented in the suffix tree by the edge it belongs to and an index within the corresponding path.
  We use $\lab(v)$ to denote the \textit{path-label} of a node $v$, i.e., the concatenation 
  of the edge labels along the path from the root to $v$.
  The terminal node corresponding to suffix $S[i \dd n]$ is marked with the index $i$.
  Each string $P$ occurring in $S$ is uniquely represented by either an explicit or an implicit node of $T$, called the \emph{locus} of $P$.
  The \textit{suffix link} of a node $v$ with path-label $\lab(v)= c P$ is a pointer to the node path-labelled $P$, 
  where $c \in \Sigma$ is a single letter and $P$ is a string. The suffix link of every non-root explicit $v$ leads to an explicit node of $T$.

  The suffix tree of a string of length $n$ even over an integer alphabet can be constructed in $\Oh(n)$ time \cite{DBLP:conf/focs/Farach97}.
  
  \subsection{Property Suffix Tree Construction}
  In analogy to the suffix tree, given a string $S$ with property $\Pi$ represented by an array $\pi$, 
  we define the \emph{property suffix tree} of $(S,\pi)$ as the compact trie representing strings $S[i\dd \pi[i]]$.
  Each terminal node $v$ stores a list $L_v$ containing all indices $i$ such that $S[i\dd \pi[i]]$ is the path-label of $v$.
  This way, $\OOcc_{\pi}(P,S)$ can be retrieved by locating the locus of $P$
  and writing down indices in lists $L_v$ for all descendants of the locus.

  For a given string $S$, we construct the property suffix tree with respect to property $\Pi$ from the suffix tree of $S$.
  This process is implemented in three steps. First, for each index $i$ we determine the locus $v_i$ of $S[i\dd \pi[i]]$.
  Next, we make all these loci explicit to create new terminal nodes. Finally, we remove nodes which should no longer exist in the tree or no longer be explicit.
  
  Our approach to the first phase is similar to Ukkonen's suffix tree construction \cite{DBLP:journals/algorithmica/Ukkonen95}.
  We are to determine the locus $v_{i}$ of $S[i \dd \pi[i]]$.
 For this, we shall traverse the suffix tree starting from an explicit node $u_i$ guaranteed to be an ancestor of $v_i$.
 We obtain $u_i$ by following the suffix link of the nearest explicit ancestor of $v_{i-1}$ ($v_{i-1}$ itself if it is explicit).
 If $i=1$ or the explicit ancestor of $v_{i-1}$ is the root, we simply set $u_i$ as the root.
 Since $\pi[i]\ge \pi[i-1]$ for $i>1$, $u_i$ is indeed an ancestor of $v_{i}$.
 Therefore, we can progress down the edges in the suffix tree from $u_i$, keeping track of the current depth until the desired depth is reached.
We know that $v_i$ exists in the tree, so it suffices to read only the first letters of each traversed edge.

This procedure results in the sequence of loci $(v_i)_{i=1}^n$.
Let us analyse its time complexity.
In the $i$-th iteration we traverse: one edge to reach $u_i$, then several edges a node $w$ whose suffix link is $u_{i+1}$, and finally at most one edge
to reach $v_i$.
Hence, the number of edges traversed in this iteration is at most $2+|\lab(u')|-|\lab(u_{i})| \le 3+|\lab(u_{i+1})|-|\lab(u_{i})|$, which gives $\Oh(n)$ overall.

The remaining steps of the algorithm are performed as follows.
We sort the loci $v_i$ by the path label length $\pi[i]-i+1$ and group them based on the edge where they are located.
This lets us appropriately subdivide each edge and compute the lists $L_v$ for the new terminal nodes.
Finally, we trim the tree: we traverse the tree bottom-up and remove or dissolve nodes which should no longer be explicit.
These steps clearly work in $\Oh(n)$ time.
  
\begin{theorem}\label{thm:prp}
For a string $S$ and property $\Pi$ represented with a table $\pi$,
the property suffix tree can be computed in $\Oh(n)$ time. Moreover, this data structure
can answer property indexing queries in $\Oh(|P|)$ time (counting) or $\Oh(|P|+|\OOcc_{\pi}(P,S)|)$ time (reporting).
\end{theorem}

\section{Weighted Index}\label{sec:appl}
Let us first describe our data structure for the Weighted Indexing problem.
For a weighted sequence $X$ and a threshold $z$, we construct a $z$-estimation $\SS=(S_j,\pi_j)_{j=1}^{\floor{z}}$ of $X$,
concatenate all the strings and shift the properties so that a single string $S$ with property $\pi$ is obtained.
Our weighted index is the property suffix tree of $S$ and $\pi$.
In the property suffix tree, each terminal node is labelled by the list of all the occurrences of the corresponding string in $S$
respecting its property.
We shift these indices so that they describe the indices within the respective strings $S_j$.

The space complexity of the index is obviously $\Oh(nz)$, where $n$ is the length of $X$.
Theorems~\ref{thm:specweight} and~\ref{thm:prp} show that the data structure can be constructed in $\Oh(nz)$ time.
The resulting weighted index is very similar to the one constructed in \cite{DBLP:conf/cpm/BartonKPR16},
even though the construction algorithm is very different.

By Definition~\ref{def:estim}, a string $P$ occurs at position $i$ in $X$ if and only if
it occurs at this position in at least one of the strings.
Thus, to check if $\Occ(P,X) \ne \emptyset$, it suffices to traverse down the property suffix tree
and check if it contains a node $v$ corresponding to $P$.
This search takes $\Oh(m)$ time, where $m=|P|$.
The two remaining types of operations---counting and reporting---require finding distinct positions
in the labels of the terminals in the subtree of $v$.
They can be implemented after additional preprocessing for the colour set size \cite{DBLP:conf/cpm/Hui92}
and coloured range listing problem \cite{Muthukrishnan2002}; details can be found in \cite{DBLP:conf/cpm/BartonKPR16}.
We obtain the same complexities as in Theorem~16 from \cite{DBLP:conf/cpm/BartonKPR16}.

\begin{theorem}\label{thm:1}
  For a weighted sequence $X$ of length $n$ over a constant-sized alphabet and a threshold $z$,
  there is a weighted index of $\Oh(nz)$ size that can be constructed in $\Oh(nz)$ time and answers decision and counting queries in $\Oh(m)$ time
  and reporting queries in $\Oh(m+|\Occ(P,X)|)$ time.
\end{theorem}

Other applications of the weighted index mentioned in \cite{DBLP:conf/cpm/BartonKPR16} include $\Oh(nz)$-time computation of the weighted prefix table
and of all covers of a weighted sequence.
Our weighted index can be used in both.

\section{Approximate Weighted Index}\label{sec:appl2}
Now let us proceed to the solution of the Approximate Weighted Indexing problem.
We are to answer queries for a pattern $P$ and a probability threshold $\frac{1}{z^{\prime}}$
and are allowed to report occurrences with probability $\ge \frac{1}{z^{\prime}}-\epsilon$,
for a given value of $\epsilon>0$.
Let us recall that \cite{DBLP:conf/edbt/BiswasPTS16} solve this problem in $\Oh(\frac{1}{\epsilon}nz^2)$ space
(with $\Omega(\frac{1}{\epsilon}n^2z^2)$ construction time)
with $\Oh(m + |\Occpar{\frac{1}{z'}-\epsilon}(P,X)|)$-time queries, assuming that $z^{\prime} \le z$ holds in all queries.
Our techniques lead to a substantial improvement over the complexities of this index.

Assume that the query is for a pattern $P$ and a threshold $\frac{1}{z^{\prime}}$.
If $\frac{1}{z^{\prime}}<\epsilon$, then the query is trivial as all the positions in $X$
can be reported.
Henceforth, we assume that $\frac{1}{z^{\prime}}\ge\epsilon$.

Let us consider a $z$-estimation $\SS$ for the weighted sequence with $z=\frac{1}{\epsilon}$.
Let $\ell=\floor{\frac{z}{z^{\prime}}}$.
By Definition~\ref{def:estim}, we can return position $i$ as an occurrence of $P$ based on
whether $\Count_{\SS}(P,i) \ge \ell$; this is shown in the following lemma.

\begin{lemma}
  If $\Count_{\SS}(P,i) \ge \ell$, then $\P_X(P,i) \ge \frac{1}{z^{\prime}}-\epsilon$.
  If $\Count_{\SS}(P,i) < \ell$, then $\P_X(P,i) < \frac{1}{z^{\prime}}$.
\end{lemma}
\begin{proof}
  Assume that $\Count_{\SS}(P,i) \ge \ell$.
  Then
  $$\P_X(P,i) \ge \tfrac1z \Count_{\SS}(P,i) \ge \tfrac1z \floor{\tfrac{z}{z^{\prime}}} \ge \tfrac1z (\tfrac{z}{z^{\prime}}-1) = \tfrac{1}{z^{\prime}}-\epsilon.$$
  Now assume that $\Count_{\SS}(P,i) < \ell$.
  As $\Count_{\SS}(P,i) = \floor{\P_X(P,i) z}$, this concludes that $\P_X(P,i) z < \ell$, which is equivalent to
  $\P_X(P,i) < \tfrac{\ell}{z} = \tfrac1z \floor{\tfrac{z}{z^{\prime}}} \le \tfrac{1}{z^{\prime}}$.
\end{proof}

Thus our approximate weighted index for $X$ is the weighted index for $X$ constructed for $z=\frac{1}{\epsilon}$.
To obtain the desired accuracy, it suffices to find the node $v$ in the property suffix tree
that corresponds to $P$ and report all positions $i$ in $X$ such that there are at least $\floor{\frac{z}{z'}}$
leaves in the subtree of $v$ labelled with the position $i$.
Let us show that this can be done by augmenting the weighted index by a data structure for \emph{(top-$k$) document retrieval}.

A version of the document retrieval problem (see Section 4.1 in \cite{DBLP:conf/soda/NavarroN12}) can be stated operationally as follows.
We are given a compact trie $T$ with $N$ leaves, each leaf labelled with a document number being a positive integer up to $N$.
(Usually $T$ is a suffix tree of a collection of documents.)
Given a pattern $P$, let $v$ be the locus of $P$.
Our goal is to report subsequent documents whose numbers occur most frequently
in the leaves of the subtree of $v$ until the process of reporting is interrupted.
In \cite{DBLP:conf/soda/NavarroN12} a data structure of size $\Oh(N)$ is shown that, given
the node $v$, reports $k$ top-scoring documents in $\Oh(k)$ time.
The construction time of the data structure is $\Oh(N \log N)$.

We can augment our property suffix tree with this data structure
with the document numbers being the labels of terminals (we can create a separate leaf for each label).
This gives $N=\Oh(nz)=\Oh(\frac{n}{\epsilon})$.
To find the documents with at least $\ell$ occurrences, we compute by doubling the smallest $k$
such that the last of the top $k$ documents reported has less than $\ell$ occurrences.
The number of documents reported in the last step of the doubling search will be at most $2|\Occpar{\frac{1}{z'}-\epsilon}(P,X)|$
and the total number will not exceed $4|\Occpar{\frac{1}{z'}-\epsilon}(P,X)|$.

\begin{theorem}\label{thm:3}
  For a weighted sequence of length $n$ over a constant-sized alphabet and parameter $\epsilon>0$,
  the Approximate Weighted Indexing problem can be solved in $\Oh(\frac{n}{\epsilon})$ space
  with $\Oh(m + |\Occpar{\frac{1}{z'}-\epsilon}(P,X)|)$-time queries.
  The construction time is $\Oh(\frac{n}{\epsilon} \log \frac{n}{\epsilon})$.
\end{theorem}

\section{Randomised Construction with Greater Space Usage}\label{sec:prob}
A symbol $X[i]$ of a weighted sequence $X$ can be interpreted as a probability distribution on $\Sigma$, 
and the whole sequence $X$ can be interpreted as a product distribution on strings of length $n$ over $\Sigma$.
In this setting, if $S \sim X$, i.e., $S$ is a random string with distribution $X$, 
then, for any position $i$ and string $P$, we have $\Pr[S[i\dd i+|P|-1]=P]=\P_X(P,i)$.
This interpretation can be used to provide a randomised construction of families $\SS$
of strings with properties equivalent to the weighted sequence $X$ in a certain sense, weaker than the one
used in Definition~\ref{def:estim}.

\begin{lemma}\label{lem:rand1}
There is a randomised algorithm which, given a weighted sequence $X$ of length $n$ and a threshold parameter $z$,
in $\Oh(nz\log(nz))$ time constructs a family $\SS$ of $k=\Oh(z \log(nz))$
strings $S_j$ with properties $\pi_j$ such that $\Count_\SS(P,i)>0$ if and only if $\P_X(P,i) \ge \frac1z$.
It succeeds with high probability ($1-\frac{1}{(nz)^c}$ for arbitrarily large constant $c$).
\end{lemma}
\begin{proof}
We randomly sample $k=\ceil{(c+2)z\ln(nz)}$ strings $S_1,\ldots,S_k$. Formally, these are independent random variables with distribution $X$.
The properties $\pi_j$ are specified so that $S_j[i\dd \pi_{j}[i]]$ is the longest prefix of $S_{j}[i\dd n]$
with $\P_X(S_j[i\dd \pi_{j}[i]], i)\ge \frac1z$.

This way, $\Count_\SS(P,i)>0$ implies $\P_X(P,i) \ge \frac1z$.
On the other hand, if $\P_X(P,i) \ge \frac1z$, then, since
$\Pr[S_j[i\dd i+|P|-1] \ne P] = 1 - \P_X(P,i)$, we have:
$$\Pr[\Count_\SS(P,i) = 0] =(1 - \P_X(P,i))^k\le e^{-k\P_X(P,i)}\le e^{-(c+2)\ln(nz)}=\tfrac{1}{(nz)^{c+2}}.$$
There are at most $n^2z$ pairs $(P,i)$ satisfying $\P_X(P,i)\ge \frac1z$
(this is the bound for the sum of lengths of all strings in the sets $\MS_i$ from Section~\ref{sec:z}).
Consequently, the resulting family has the required property with 
probability at least $1-\frac{n^2z}{(nz)^{c+2}}\ge 1-\frac{1}{(nz)^c}$.
\end{proof}

We can directly use the same methods as in Section~\ref{sec:appl} to construct a weighted index
from the family of strings constructed in Lemma~\ref{lem:rand1}.
The space complexity of the resulting index is worse than the one in Theorem~\ref{thm:1} by a factor of $\log(nz)$
and the construction is randomised.

\begin{corollary}
There is a data structure of size $\Oh(nz\log(nz))$ for the Weighted Indexing problem which answers queries in optimal time.
It can be constructed using a randomised $\Oh(nz\log(nz))$-time algorithm which returns
a valid weighted index with high probability.
\end{corollary}

The same type of construction can be used to obtain an approximate weighted index.
To this end, we need a stronger equivalence property of a string family
and a greater number of sampled strings to satisfy this property.

\begin{lemma}
There is a randomised algorithm which, given a weighted sequence $X$ of length $n$ and a parameter $\epsilon$,
in $\Oh(\frac{n}{\epsilon^2}\log(\frac{n}{\epsilon}))$ time constructs a family $\SS$ of $k=\Oh(\frac{1}{\epsilon^2}\log(\frac{n}{\epsilon}))$
strings $S_j$ with properties $\pi_j$ such that $|\P_X(P,i)-\frac1k\Count_\SS(P,i)|<\epsilon$ for every position $i$ and string $P$.
It succeeds with high probability ($1-(\frac{\epsilon}{n})^c$ for arbitrarily large constant $c$).
\end{lemma}
\begin{proof}
We randomly sample $k=\ceil{(c+2)\frac1{\epsilon^2}\ln\frac{n}{\epsilon}}$ strings $S_1,\ldots,S_k$.
The properties $\pi_j$ satisfy that $S_j[i\dd \pi_{j}[i]]$ is the longest prefix of $S_{j}[i\dd n]$
such that $\P_X(S_j[i\dd \pi_{j}[i]], i)\ge \epsilon$.

Observe that if $\P_X(P,i)< \epsilon$, then $\Count_{\SS}(P,i)=0$.
On the other hand, if $\P_X(P,i)\ge \epsilon$,
then $\Count_{\SS}(P,i) \sim \mathrm{Bin}(k, \P_X(P,i))$.
Consequently, Hoeffding's inequality \cite{doi:10.1080/01621459.1963.10500830} implies
$$\Pr[| \P_X(P,i)- \tfrac{1}{k}\Count_{\SS}(P,i)| > \epsilon] \le 2e^{-\epsilon^2 k}=2e^{-(c+2)\ln\tfrac{n}{\epsilon}}\le 2(\tfrac{\epsilon}{n})^{c+2}.$$
There are at most $\frac{n^2}{\epsilon}$ such pairs $(P,i)$, so the family $\SS$ satisfies
the required condition with probability at least
$1-(\frac{\epsilon}{n})^c$, as claimed.
\end{proof}

We can use this family of strings to construct an approximate weighted index
using top-$k$ document retrieval just as in Section~\ref{sec:appl2}.
We arrive at the following construction with space complexity greater than the one from Theorem~\ref{thm:3}
by a factor of $\frac{1}{\epsilon} \log \frac{n}{\epsilon}$ (and has a randomised construction).

\begin{corollary}
There is a data structure of size $\Oh(\frac{n}{\epsilon^2}\log\frac{n}{\epsilon})$ which solves the Approximate Weighted Indexing problem with $\Oh(m + |\Occpar{\frac{1}{z'}-\epsilon}(P,X)|)$-time queries.
It can be constructed using a randomised $\Oh(\frac{n}{\epsilon^2}\log^2\frac{n}{\epsilon})$-time algorithm which returns
a valid approximate weighted index with high probability.
\end{corollary}

\section{Conclusions}\label{sec:concl}
In this article we present an efficient index for Weighted Pattern Matching
along with new combinatorial insights into the nature of weighted sequences.
We have produced an implementation of the index (see \texttt{https://bitbucket.org/kociumaka/weighted\_index}) that we have validated for correctness and efficiency against
known weighted pattern matching algorithms \cite{DBLP:conf/isaac/KociumakaPR16,DBLP:journals/corr/BartonLP15,DBLP:conf/cocoa/BartonLP16}.
Our implementation supports decision, counting, and reporting variants of queries; however, only decision operations were implemented in worst-case optimal time.

Let us mention that our results can be extended to integer alphabets $\Sigma$, i.e., $\Sigma \subseteq \{1,\dots,n^{\Oh(1)}\}$, without influencing the space
and construction time.
We have omitted the description of this extension and preferred to focus on the basic case of a constant-sized alphabet
that is also most relevant in practice.

Finally, our ideas can be used to improve the solution for the Generalised Weighted Indexing problem from \cite{DBLP:conf/edbt/BiswasPTS16}.
They use a notion of \emph{special weighted sequences}
in which each position contains at most one letter with a positive probability.
(In this case the assumption that the probabilities sum up to 1 at each position is waived.)
In \cite{DBLP:conf/edbt/BiswasPTS16} the input weighted sequence is transformed using the reduction of~\cite{amir_weighted_property_matching_j}
into a special weighted sequence of length $\Oh(nz^2 \log z)$ that preserves the set of maximal solid factors.
In the special weighted sequence, a query for a pattern $P$ under the probability threshold $\frac{1}{z^{\prime}}$ is answered
in $\Oh(m + m \cdot |\Occpar{\frac{1}{z^{\prime}}}(P,X)|)$ time.

Our $z$-estimation $\SS$ can be transformed into a special weighted sequence of length $\Oh(nz)$ that also preserves the
set of solid factors.
We simply concatenate the strings, taking the letter probabilities from the respective positions in $X$, and split the concatenated parts
with a zero-probability position.
This gives a more space-efficient reduction that can be used in the data structure of \cite{DBLP:conf/edbt/BiswasPTS16}.

\begin{corollary}\label{thm:2}
  For a weighted sequence of length $n$ over an integer alphabet, the Generalised Weighted Indexing problem
  can be solved with $\Oh(m + m \cdot |\Occpar{\frac{1}{z'}}(P,X)|)$-time queries with an index of size $\Oh(nz)$.
\end{corollary}

\subsection*{Acknowledgement}
We thank an anonymous referee of the previous version of the paper for the idea of a simple randomised construction.
We also thank Tsvi Kopelowitz for bringing our attention to the multitude of existing solutions to the Property Indexing problem.

\bibliographystyle{plain}
\bibliography{better_weighted_index}

\begin{thebibliography}{10}

\bibitem{DBLP:journals/tkde/AggarwalY09}
Charu~C. Aggarwal and Philip~S. Yu.
\newblock A survey of uncertain data algorithms and applications.
\newblock {\em {IEEE} Transactions on Knowledge and Data Engineering},
  21(5):609--623, 2009.

\bibitem{amir_weighted_property_matching_j}
Amihood Amir, Eran Chencinski, Costas~S. Iliopoulos, Tsvi Kopelowitz, and Hui
  Zhang.
\newblock Property matching and weighted matching.
\newblock {\em Theoretical Computer Science}, 395(2-3):298--310, April 2008.

\bibitem{DBLP:conf/cpm/BartonKPR16}
Carl Barton, Tomasz Kociumaka, Solon~P. Pissis, and Jakub Radoszewski.
\newblock Efficient index for weighted sequences.
\newblock In Roberto Grossi and Moshe Lewenstein, editors, {\em Combinatorial
  Pattern Matching, {CPM} 2016}, volume~54 of {\em LIPIcs}, pages 4:1--4:13,
  Dagstuhl, Germany, 2016. Schloss Dagstuhl--Leibniz-Zentrum für Informatik.

\bibitem{DBLP:journals/corr/BartonLP15}
Carl Barton, Chang Liu, and Solon~P. Pissis.
\newblock Fast average-case pattern matching on weighted sequences, 2015.

\bibitem{DBLP:conf/cocoa/BartonLP16}
Carl Barton, Chang Liu, and Solon~P. Pissis.
\newblock On-line pattern matching on uncertain sequences and applications.
\newblock In T.{-}H.~Hubert Chan, Minming Li, and Lusheng Wang, editors, {\em
  Combinatorial Optimization and Applications, {COCOA} 2016}, volume 10043 of
  {\em LNCS}, pages 547--562. Springer, 2016.

\bibitem{DBLP:conf/edbt/BiswasPTS16}
Sudip Biswas, Manish Patil, Sharma~V. Thankachan, and Rahul Shah.
\newblock Probabilistic threshold indexing for uncertain strings.
\newblock In Evaggelia Pitoura, Sofian Maabout, Georgia Koutrika, Am{\'{e}}lie
  Marian, Letizia Tanca, Ioana Manolescu, and Kostas Stefanidis, editors, {\em
  19th International Conference on Extending Database Technology, {EDBT} 2016},
  pages 401--412. OpenProceedings.org, 2016.

\bibitem{KCL_publication}
Manolis Christodoulakis, Costas~S. Iliopoulos, Laurent Mouchard, and Kostas
  Tsichlas.
\newblock Pattern matching on weighted sequences.
\newblock In {\em Algorithms and Computational Methods for Biochemical and
  Evolutionary Networks, CompBioNets 2004}, KCL publications, 2004.

\bibitem{DBLP:journals/cacm/DalviRS09}
Nilesh~N. Dalvi, Christopher R{\'{e}}, and Dan Suciu.
\newblock Probabilistic databases: diamonds in the dirt.
\newblock {\em Communications of the {ACM}}, 52(7):86--94, 2009.

\bibitem{DBLP:conf/focs/Farach97}
Martin Farach.
\newblock Optimal suffix tree construction with large alphabets.
\newblock In {\em 38th {IEEE} Annual Symposium on Foundations of Computer
  Science, {FOCS} 1997}, pages 137--143. {IEEE} Computer Society, 1997.

\bibitem{doi:10.1080/01621459.1963.10500830}
Wassily Hoeffding.
\newblock Probability inequalities for sums of bounded random variables.
\newblock {\em Journal of the American Statistical Association},
  58(301):13--30, 1963.

\bibitem{DBLP:conf/cpm/Hui92}
Lucas Chi~Kwong Hui.
\newblock Color set size problem with application to string matching.
\newblock In Alberto Apostolico, Maxime Crochemore, Zvi Galil, and Udi Manber,
  editors, {\em Combinatorial Pattern Matching, {CPM} 1992}, volume 644 of {\em
  LNCS}, pages 230--243. Springer, 1992.

\bibitem{costas_weighted_suffix_tree_j}
Costas~S. Iliopoulos, Christos Makris, Yannis Panagis, Katerina Perdikuri,
  Evangelos Theodoridis, and Athanasios~K. Tsakalidis.
\newblock The weighted suffix tree: An efficient data structure for handling
  molecular weighted sequences and its applications.
\newblock {\em Fundamenta Informaticae}, 71(2-3):259--277, 2006.

\bibitem{DBLP:journals/ipl/IliopoulosR08}
Costas~S. Iliopoulos and Mohammad~Sohel Rahman.
\newblock Faster index for property matching.
\newblock {\em Information Processing Letters}, 105(6):218--223, 2008.

\bibitem{DBLP:journals/ipl/JuanLW09}
M.~T. Juan, J.~J. Liu, and Y.~L. Wang.
\newblock Errata for ``faster index for property matching''.
\newblock {\em Information Processing Letters}, 109(18):1027--1029, 2009.

\bibitem{DBLP:conf/isaac/KociumakaPR16}
Tomasz Kociumaka, Solon~P. Pissis, and Jakub Radoszewski.
\newblock Pattern matching and consensus problems on weighted sequences and
  profiles.
\newblock In Seok{-}Hee Hong, editor, {\em Algorithms and Computation, {ISAAC}
  2016}, volume~64 of {\em LIPIcs}, pages 46:1--46:12. Schloss
  Dagstuhl--Leibniz-Zentrum für Informatik, 2016.

\bibitem{DBLP:journals/tcs/Kopelowitz16}
Tsvi Kopelowitz.
\newblock The property suffix tree with dynamic properties.
\newblock {\em Theor. Comput. Sci.}, 638:44--51, 2016.

\bibitem{Muthukrishnan2002}
S.~Muthukrishnan.
\newblock Efficient algorithms for document retrieval problems.
\newblock In David Eppstein, editor, {\em 13th Annual {ACM-SIAM} Symposium on
  Discrete Algorithms, {SODA 2002}}, pages 657--666. {ACM/SIAM}, 2002.

\bibitem{DBLP:conf/soda/NavarroN12}
Gonzalo Navarro and Yakov Nekrich.
\newblock Top-$k$ document retrieval in optimal time and linear space.
\newblock In Yuval Rabani, editor, {\em 23rd Annual {ACM-SIAM} Symposium on
  Discrete Algorithms, {SODA} 2012}, pages 1066--1077. {SIAM}, 2012.

\bibitem{DBLP:journals/jcb/RajasekaranJS02}
Sanguthevar Rajasekaran, X.~Jin, and John~L. Spouge.
\newblock The efficient computation of position-specific match scores with the
  fast {F}ourier transform.
\newblock {\em Journal of Computational Biology}, 9(1):23--33, 2002.

\bibitem{DBLP:conf/isaac/Shibuya99}
Tetsuo Shibuya.
\newblock Constructing the suffix tree of a tree with a large alphabet.
\newblock In Alok Aggarwal and C.~Pandu Rangan, editors, {\em Algorithms and
  Computation, {ISAAC} 1999}, volume 1741 of {\em {LNCS}}, pages 225--236.
  Springer, 1999.

\bibitem{DBLP:journals/algorithmica/Ukkonen95}
Esko Ukkonen.
\newblock On-line construction of suffix trees.
\newblock {\em Algorithmica}, 14(3):249--260, 1995.

\end{thebibliography}

\end{document}